%% file: main.tex
\documentclass[11pt]{article}
\input{preamble}

\title{Bicriteria Approximate Navigable Graphs in General Metrics}
\author{%
  Hongyi Chen\thanks{State Key Laboratory for Novel Software Technology,
  Nanjing University, \href{mailto:hychener01@gmail.com}{\texttt{hychener01@gmail.com}}}
  \and
  Tianyi Zhang\thanks{State Key Laboratory for Novel Software Technology,
  Nanjing University, \href{mailto:tianyiz25@nju.edu.cn}{\texttt{tianyiz25@nju.edu.cn}}}}
\date{}

\begin{document}
\pagenumbering{roman}
\maketitle
\begin{abstract}
\input{sections/abstract}
\end{abstract}
\thispagestyle{empty}
\clearpage
\pagenumbering{arabic}
\pagestyle{plain}

\input{intro}
\input{alg}
\input{lbnd}
\input{sections/ai_disclosure}

\input{sections/acknowledgment}

\bibliographystyle{alpha}
\bibliography{references}
\end{document}

%% file: preamble.tex
\usepackage[letterpaper,margin=1in]{geometry}
\usepackage[T1]{fontenc}
\usepackage{amsmath,amsfonts,amssymb,amsthm,mathtools}
\usepackage{mathpazo,tgpagella}
\usepackage{graphicx,xspace}
\usepackage[algo2e,ruled,vlined,linesnumbered,commentsnumbered,noend]{algorithm2e}
\SetKwInput{KwIn}{Input}
\SetKwInput{KwOut}{Output}
\SetKw{Return}{return}
\SetKwComment{Comment}{// }{}
\DontPrintSemicolon
\usepackage[shortlabels]{enumitem}
\usepackage{booktabs,array,microtype,placeins,needspace}
\usepackage[dvipsnames]{xcolor}
\usepackage{tikz,subcaption}
\usetikzlibrary{arrows.meta,trees,shapes.geometric}
\usepackage{hyperref}
\usepackage{url}
\usepackage[capitalize,nameinlink]{cleveref}
\hypersetup{colorlinks=true,pdfpagemode=UseNone,citecolor=OliveGreen,
  linkcolor=NavyBlue,urlcolor=Magenta,pdfstartview=FitH,
  pdftitle={Bicriteria Approximate Navigable Graphs in General Metrics},
  pdfauthor={Hongyi Chen and Tianyi Zhang}}
\newtheorem{theorem}{Theorem}[section]
\newtheorem{lemma}[theorem]{Lemma}
\newtheorem{claim}[theorem]{Claim}

\theoremstyle{definition}
\newtheorem{definition}[theorem]{Definition}
\theoremstyle{remark}

\newtheorem{question}[theorem]{Question}
\AddToHook{env/theorem/begin}{\crefalias{section}{theorem}\crefalias{theorem}{theorem}}
\AddToHook{env/lemma/begin}{\crefalias{section}{lemma}\crefalias{theorem}{lemma}}
\AddToHook{env/claim/begin}{\crefalias{section}{claim}\crefalias{theorem}{claim}}
\AddToHook{env/proposition/begin}{\crefalias{section}{proposition}\crefalias{theorem}{proposition}}
\AddToHook{env/corollary/begin}{\crefalias{section}{corollary}\crefalias{theorem}{corollary}}
\AddToHook{env/definition/begin}{\crefalias{section}{definition}\crefalias{theorem}{definition}}
\AddToHook{env/remark/begin}{\crefalias{section}{remark}\crefalias{theorem}{remark}}
\AddToHook{env/question/begin}{\crefalias{section}{question}\crefalias{theorem}{question}}
\crefname{question}{Question}{Questions}
\Crefname{question}{Question}{Questions}
\crefname{algocf}{Algorithm}{Algorithms}
\Crefname{algocf}{Algorithm}{Algorithms}
\numberwithin{equation}{section}
\numberwithin{algocf}{section}
\setlist{itemsep=3pt,topsep=4pt}

\AddToHook{cmd/section/before}{\Needspace{5\baselineskip}}
\AddToHook{cmd/subsection/before}{\Needspace{4\baselineskip}}
\input{macros}

%% file: macros.tex
        {\medskip}

\newcommand{\ceil}[1]{\left\lceil #1 \right\rceil}

\newcommand{\brac}[1]{\left(#1\right)}

\newcommand{\sfd}{\mathsf{d}}

%% file: sections/abstract.tex
We study approximation algorithms for constructing sparse navigable graphs in general metrics. 
Given an $n$-point metric $(P,\sfd)$ and a parameter $\alpha\ge 1$, an $\alpha$-navigable graph is a directed graph $G=(P,E)$ such that for every ordered pair of distinct points $s,t\in P$, the vertex $s$ has an outgoing edge $(s,p)$ with
\(
    \sfd(p,t)<\sfd(s,t)/\alpha .
\) The objective is to minimize the number of edges.
The study of approximation algorithms for sparse navigable graphs in general metrics was recently initiated by Khanna, Padaki, and Waingarten (SODA~2026), and independently by Conway et al.\ (SODA~2026). These works show that the single-criterion problem faces a Set Cover barrier: an $O(\log n)$-approximation is achievable, whereas obtaining an $o(\log n)$-approximation is NP-hard. To obtain faster algorithms, Khanna, Padaki, and Waingarten (SODA~2026) introduced bicriteria approximation, in which the output graph is required to be $\alpha$-navigable, while its size is compared against that of an optimal $\beta$-navigable graph for some relaxed parameter $\beta \ge \alpha$. They gave a fast matrix-multiplication-based algorithm that computes an $(\alpha,2(1+\varepsilon)\alpha)$-bicriteria $O(\log n)$-approximation in time $\widetilde O(n^\omega \log \Delta/\varepsilon)$, where $\Delta$ denotes the aspect ratio and $\omega$ is the matrix multiplication exponent.
Our main result is a simpler near-quadratic-time algorithm with stronger bicriteria guarantees. For every constant $\alpha\ge 1$, we compute in time $\widetilde O(n^2\log\Delta)$ an $\alpha$-navigable graph that is simultaneously an $(\alpha,2\alpha+O(1))$-bicriteria $O(1)$-approximation and an $(\alpha,\alpha+O(1))$-bicriteria $O(\log n)$-approximation. We complement the algorithm with a hardness result showing that for every constant $\eta\in(0,1)$, no polynomial-time algorithm can compute an $(\alpha,2\alpha-\eta)$-bicriteria $o(\log n)$-approximation unless $\mathrm{P}=\mathrm{NP}$.

%% file: intro.tex
\section{Introduction.}
Nearest neighbor search has been a fundamental problem in computer science \cite{bentley1975multidimensional,kushilevitz1998efficient,indyk1998approximate,jegou2010product,Malkov2016EfficientAR}, with wide applications in modern AI systems \cite{johnson2019billion, lewis2020retrieval,gao2023retrieval}. In this problem, we are given a point set $P$ together with a pairwise distance function $\sfd(\cdot, \cdot)$, and the goal is to build an efficient data structure which (approximately) answers $\arg\min_{p\in P}\sfd(p,q)$ for a query point \(q\) in the underlying metric space. Locality-sensitive hashing has been a classical approach with abundant theoretical studies for this problem \cite{Andoni2015PracticalAO}, but it becomes inefficient for large high-dimensional datasets. A recent line of practical work has adopted graph-based paradigms, including HNSW \cite{Malkov2016EfficientAR}, NSG \cite{Fu2017FastAN}, and DiskANN~\cite{Subramanya2019DiskANNF}; these heuristic approaches have yielded impressive empirical performance on real-world datasets \cite{10.14778/3476249.3476255,azizi2025graph}. 
A central theoretical abstraction behind graph-based search is \textit{navigability}, which is defined as follows.
\begin{definition}[\cite{Subramanya2019DiskANNF,10.5555/3666122.3669013}]
	Given a constant parameter $\alpha\geq 1$, for any ordered pair $(s,t)\in P\times P$ with $s\ne t$ and any point $u\in P$, we say that a directed edge $(s,u)$ is an \emph{$\alpha$-navigation from $s$ to $t$} if $$\sfd(u,t) < \frac{\sfd(s,t)}{\alpha}.$$
    A directed graph $G = (P,E)$ is said to be \emph{$\alpha$-navigable} if for every ordered pair $(s,t)\in P\times P$ with $s\ne t$, the vertex $s$ has at least one outgoing edge that forms an $\alpha$-navigation toward $t$.
\end{definition}

Graph-based nearest-neighbor search methods typically proceed in two stages. In the preprocessing stage, one builds a navigable graph \(G=(P,E)\) on the dataset. At query time, many such methods perform a greedy or greedy-search-like traversal over \(G\): starting from one or several entry points, they repeatedly inspect out-neighbors of currently active vertices and try to move the search toward points closer to the query. Thus, the quality of the graph is governed by two competing requirements: the graph should contain enough edges to satisfy the navigation constraints, but it should also be sparse so that memory usage and the cost of inspecting neighbors remain small. It was shown in \cite{10.5555/3666122.3669013} that, for any $\alpha>1$, any \(\alpha\)-navigable graph can be used to construct a data structure which returns $\left(\frac{\alpha+1}{\alpha-1}+\epsilon\right) $-approximate nearest neighbors in $O\left(\log_\alpha \frac{\Delta}{(\alpha-1)\epsilon}\right)$
greedy steps, where $\Delta$ is the aspect ratio of the metric, namely the ratio of the maximum to the minimum distance between distinct points of $P$.

This motivates the following problem: given an \(n\)-point metric space \((P,\sfd)\) and a parameter \(\alpha\ge 1\), construct an \(\alpha\)-navigable graph using as few edges as possible. The authors of~\cite{diwan2024navigable} showed that every set of \(n\) distinct points admits a $1$-navigable graph with \(\widetilde O(n^{3/2})\)\footnote{$\tilde{O}$ hides poly-logarithmic factors.} edges, and gave a \(\widetilde O(n^2)\)-time construction. In a different direction, Arya and Mount~\cite{arya1993approximate} showed that, for point sets in $d$-dimensional Euclidean space, one can construct a $1$-navigable graph with maximum degree \(2^{O(d)}\). However, they are not instance-optimal: a particular input \(P\) may admit a much sparser navigable graph than the worst-case bound.

Two recent concurrent works \cite{Khanna2025SparseNG,Conway2025EfficientlyCS} have initiated the study of approximation algorithms for constructing sparse $\alpha$-navigable graphs in general metrics. In both works, the authors showed a reduction from Set Cover to $\alpha$-navigable graphs, establishing NP-hardness of $o(\log n)$-approximations of the sparsest $\alpha$-navigable graphs. On the algorithmic side, both works designed an $\tilde{O}\brac{\min\{n \cdot \mathrm{OPT}_\alpha, n^{2.5}\}}$ time algorithm which computes $O(\log n)$-approximations of the sparsest $\alpha$-navigable graphs, where $\mathrm{OPT}_\alpha$ is the number of edges in the optimal solution. Besides, for the special case of $\alpha = 1$, the authors of \cite{Conway2025EfficientlyCS} showed a $\tilde{O}(n^2)$ time algorithm with $O(\log n)$-approximation. For general values of $\alpha \geq 1$, to achieve faster runtime, the authors of \cite{Khanna2025SparseNG} have studied the notion of bicriteria approximation of $\alpha$-navigable graphs defined below.
\begin{definition}[\cite{Khanna2025SparseNG}]\label{def:bicriteria}
	Let $(P,\sfd)$ be a finite metric space and $\beta \ge \alpha \ge 1$.
	A directed graph $G = (P,E)$ is a \emph{$(\alpha,\beta)$-bicriteria $c$-approximation} to the
	sparsest $\alpha$-navigable graph if $G$ is $\alpha$-navigable and 
    $$
    |E| \le c \cdot \min_{H\ \text{is $\beta$-navigable}} |E(H)|.
    $$
\end{definition}
To unify notation, the single criterion problem of approximating $\alpha$-navigable graphs can be written as a $(\alpha, \alpha)$-bicriteria problem. In \cite{Khanna2025SparseNG}, it was shown that a $(\alpha, 2(1+\epsilon)\alpha)$-bicriteria $O(\log n)$-approximation can be computed by an algorithm in $\tilde{O}(n^\omega\log\Delta / \epsilon)$ time, where $\Delta$ is the aspect ratio of metric $(P, \sfd)$ and $\omega$ is the matrix multiplication exponent. The following question was left open in \cite{Khanna2025SparseNG}.

\begin{question}\label{question}
Is there a quadratic-time algorithm for (bicriteria)
$O(\log n)$-approximation of the sparsest
$\alpha$-navigable graph?
\end{question}

\subsection{Our Results.}
As our main algorithmic result, we answer \Cref{question} affirmatively up to polylogarithmic factors in the bicriteria setting. More precisely, we give a near-quadratic-time construction whose bicriteria guarantees improve upon the previous fastest algorithm and are essentially tight in the large-$\alpha$ regime.

\begin{theorem}\label{thm:alg}
	Let $(P, \sfd)$ be an $n$-point metric space with aspect ratio $\Delta$ and $\alpha\geq 1$ be a constant parameter. Then, there is an algorithm with runtime $\tilde{O}(n^2\log\Delta)$ which builds an $\alpha$-navigable graph $G = (P, E)$ with the following bicriteria approximation guarantees simultaneously.
	\begin{itemize}[leftmargin=*]
		\item $G$ is an $(\alpha, \beta)$-bicriteria $5$-approximation, where $\beta = 2\alpha + 8 + 4/\alpha \leq 2\alpha + 12$;
		\item $G$ is an $(\alpha, \beta')$-bicriteria $O(\ln n)$-approximation, where $\beta' = \alpha + 5 + 2/\alpha \leq \alpha + 7$.
	\end{itemize}
\end{theorem}

The first guarantee shows that the Set Cover barrier can be bypassed by allowing an almost factor-$2$ relaxation in the bicriteria parameter: instead of an $O(\log n)$-approximation, we obtain
a constant-factor approximation against the sparsest $(2\alpha + O(1))$-navigable graph. The second guarantee gives an $O(\log n)$-approximation with a smaller bicriteria loss
than the previous fast bicriteria algorithm of \cite{Khanna2025SparseNG}; in particular, for large constant $\alpha$, the leading coefficient improves from roughly $2\alpha$ to roughly $\alpha$. Another main contribution is the running time: both guarantees are achieved in \(\widetilde O(n^2\log\Delta)\) time, improving over the \(\widetilde O(n^\omega\log\Delta/\epsilon)\)-time fast bicriteria algorithm of \cite{Khanna2025SparseNG}. Since the input metric is specified by \(\Theta(n^2)\) pairwise distances, this running time is optimal up to logarithmic factors if pairwise distances are given explicitly.

\begin{table}[!htbp]
\centering
\small
\caption{Upper and lower bounds for sparse navigable graphs in general metrics.
Here \(\mathrm{OPT}_{\alpha}\) denotes the minimum number of edges in an \(\alpha\)-navigable graph, \(\Delta\) is the aspect ratio, and \(\omega\) is the matrix multiplication exponent. We use the bicriteria notation of \Cref{def:bicriteria}.}
\label{tab:related-results}
\vspace{3pt}
\begin{tabular*}{\textwidth}{@{\extracolsep{\fill}}llll@{}}
\toprule
Guarantee & Approx. & Time & Reference \\
\midrule
\((1, 1) \)-bicriteria
& \(O(\log n)\)
& \(\widetilde O(n^2)\)
& \cite{Conway2025EfficientlyCS} \\
\((\alpha,\alpha) \)-bicriteria
& \(O(\log n)\)
& \(\widetilde O\brac{\min\{n\cdot \mathrm{OPT}_{\alpha}, n^{2.5}\}}\)
& \cite{Conway2025EfficientlyCS,Khanna2025SparseNG} \\
\((\alpha,\alpha) \)-bicriteria
& \(\Omega(\log n)\)
& Polynomial time
& \cite{Khanna2025SparseNG} \\
\((\alpha,2(1+\epsilon)\alpha)\)-bicriteria
& \(O(\log n)\)
& \(\widetilde O(n^{\omega}\log\Delta/\epsilon)\)
& \cite{Khanna2025SparseNG} \\

\((\alpha,\alpha+O(1))\)-bicriteria
& \(O(\log n)\)
& \(\widetilde O(n^2\log\Delta)\)
& This work \\

\((\alpha,2\alpha+O(1))\)-bicriteria
& \(O(1)\)
& \(\widetilde O(n^2\log\Delta)\)
& This work \\
\((\alpha, 2\alpha-\eta)\)-bicriteria    &   \(\Omega(\log n)\)  &   Polynomial time &   This work\\
\bottomrule
\end{tabular*}
\end{table}

Our second result shows that the factor $2$ appearing above is essentially unavoidable. Below $2\alpha$, the Set Cover barrier persists.
\begin{theorem}\label{thm:reduction}
	For every constant $\alpha \ge 1$ and every constant $\eta\in(0,1)$, there exists a constant $c_\eta>0$ such that it is NP-hard to compute an $(\alpha, 2\alpha - \eta)$-bicriteria $c_\eta\ln n$-approximation of the sparsest $\alpha$-navigable graph on a metric $(P, \sfd)$ with $n$ points.
\end{theorem}

Together with \Cref{thm:alg}, \Cref{thm:reduction} shows that our two upper-bound guarantees are nearly optimal in complementary ways. First, the lower bound shows that the Set Cover barrier persists for every bicriteria parameter strictly below \(2\alpha\): even if the algorithm is allowed to compare against the sparsest \((2\alpha-\eta)\)-navigable graph, an \(o(\log n)\)-approximation remains NP-hard. Therefore, in the entire sub-\(2\alpha\) bicriteria regime, the best possible approximation factor is \(\Theta(\log n)\); our \((\alpha,\alpha+O(1))\)-bicriteria \(O(\log n)\)-approximation matches this approximation factor while using bicriteria loss \(\alpha+O(1)=(1+o_\alpha(1))\alpha\), which is asymptotically the smallest possible since every bicriteria comparison must have \(\beta\ge\alpha\). Second, the same lower bound shows that any constant-factor approximation must cross the \(2\alpha\) threshold, while our \((\alpha,2\alpha+O(1))\)-bicriteria \(O(1)\)-approximation does so by only an additive constant. Thus, for large \(\alpha\), our results identify \(2\alpha\) as the sharp threshold for beating the logarithmic Set Cover barrier, and achieve essentially the smallest bicriteria loss possible on the logarithmic-approximation side of this threshold.

We emphasize that our main focus is the regime where \(\alpha\) is large. This is precisely the regime in which the standard theoretical guarantee for greedy search becomes meaningful for high-accuracy nearest-neighbor search: on an \(\alpha\)-navigable graph, greedy search achieves approximation ratio
\(
    \frac{\alpha+1}{\alpha-1}+\varepsilon.
\) Since \(\frac{\alpha+1}{\alpha-1}=1+\frac{2}{\alpha-1}\), obtaining a \((1+\epsilon)\)-approximate nearest neighbor through this analysis requires \(\alpha=\Theta(1/\epsilon)\), after choosing the additive accuracy parameter \(\varepsilon\) appropriately. In this sense, our construction gives a nearly tight bicriteria approximation in the theoretically relevant high-accuracy regime.

We note that in practice, the values of $\alpha$ used or observed in proximity graphs are often smaller than $2$ \cite{Subramanya2019DiskANNF}. However, the nearest neighbor approximation guarantee of greedy search on $\alpha$-navigable graphs for small choices of $\alpha$ is still poorly understood, and this is beyond the scope of the present work.

\subsection{Technical Overview.}
We give an overview of the main ideas behind our upper bound. We start with the Set Cover viewpoint, which explains the \(O(\log n)\) barrier in previous work. We then describe a simple but slower ordered-greedy procedure, which exposes the packing principle behind the \(2\alpha\) threshold. Our algorithm can be viewed as an acceleration of the ordered-greedy procedure, achieved by replacing the source-specific greedy choices with globally shared representatives.

\paragraph{The Set Cover viewpoint.}
For a fixed source point \(s\), the \(\alpha\)-navigability constraints can be written as a Set Cover instance.
The universe is \(P\setminus\{s\}\), and each candidate out-neighbor \(u\neq s\) corresponds to the set
\[
Z_\alpha(s,u) = \left\{ t\in P\setminus\{s\} \;\middle|\; \sfd(u,t)<\frac{\sfd(s,t)}{\alpha} \right\}.
\]
Choosing out-neighborhoods for \(s\) is therefore exactly the problem of covering all targets \(t\in P\setminus\{s\}\) by sets \(Z_\alpha(s,u)\).
Constructing a sparsest \(\alpha\)-navigable graph amounts to solving one such Set Cover instance for every source \(s\).

This formulation immediately yields \(O(\log n)\)-approximation algorithms via greedy Set Cover, and it also explains why the problem inherits the classical Set Cover hardness barrier.
Previous algorithms~\cite{Khanna2025SparseNG,Conway2025EfficientlyCS} for sparse navigable graphs in general metrics are based on this viewpoint.
However, if one wants to beat the \(O(\log n)\) approximation factor, the Set Cover formulation by itself does not reveal enough additional structure.

\paragraph{Warm-up: source-wise ordered greedy.}
Fix a source \(s\) and consider the following simple greedy procedure.
Process all targets \(t\in P\setminus\{s\}\) in nondecreasing order of \(\sfd(s,t)\).
Whenever no previously added outgoing edge of $s$ is an $\alpha$-navigation from $s$ to the current target $t$, add the edge $(s,t)$. This procedure clearly produces an \(\alpha\)-navigable graph, since any uncovered target can always be covered by selecting itself.
Let $ p_1,p_2,\ldots,p_k $ be the out-neighbors selected for source \(s\), in nondecreasing order of their distance from \(s\). When $p_j$ is selected, none of the previously added edges $(s,p_i)$, $i<j$, is an $\alpha$-navigation from $s$ to $p_j$. Therefore, \begin{equation}\label{eq:ordered-greedy-separation} \sfd(p_i,p_j) \ge \frac{\sfd(s,p_j)}{\alpha}\qquad\text{for every } i<j. \end{equation}

Now compare this ordered-greedy solution with any \(2\alpha\)-navigable graph \(H\). For every selected point \(p_j\), the graph \(H\) contains an out-neighbor \(u_j\) of \(s\) such that
\[ \sfd(u_j,p_j) < \frac{\sfd(s,p_j)}{2\alpha}. \]
Choose one such out-neighbor $u_j$ for each selected point $p_j$.
We show that the map $p_j\mapsto u_j$ is injective, which bounds
the number of selected points by the out-degree of $s$ in $H$.
Indeed, suppose two selected points $p_i,p_j$ with $i<j$
were mapped to the same out-neighbor $u$.
Then by triangle inequality,
\[
\sfd(p_i,p_j) \le \sfd(p_i,u)+\sfd(u,p_j) <
\frac{\sfd(s,p_i)}{2\alpha} + \frac{\sfd(s,p_j)}{2\alpha}
\le \frac{\sfd(s,p_j)}{\alpha},
\]
where the last inequality uses \(\sfd(s,p_i)\le \sfd(s,p_j)\).
This contradicts \eqref{eq:ordered-greedy-separation}.
Hence, for this fixed source \(s\), the ordered-greedy procedure selects no more out-neighbors than the out-degree of \(s\) in any \(2\alpha\)-navigable graph.

However, this ordered-greedy procedure is too slow for our purposes. Running it independently for all sources requires repeatedly testing whether each target is already navigated by one of the previously selected out-neighbors. A direct implementation takes \(O(n^3)\) time in the worst case.

\paragraph{Global acceleration via hierarchical covers.}
Our starting point is the classical use of nets and hierarchical covers, which has been widely adopted in low-dimensional Euclidean and doubling metrics~\cite{Krauthgamer2004NavigatingNS,Lu2025ProximityGF}. Instead of running greedy independently for every source, we precompute globally shared representatives at multiple distance scales. More concretely, let $r_i=\left(1+\frac1\alpha\right)^i$ range over the relevant distance scales. At each scale $r_i$, we use balls of radius $\Theta(r_i/\alpha)$ and, for each source $s$, add edges to selected centers lying in an appropriate annulus around $s$. The annulus is chosen so that if a target $t$ satisfies $\sfd(s,t)\approx r_i$, and $p$ is the center covering $t$ at scale $i$, then $p$ is included as an out-neighbor of $s$ and
\[
\sfd(p,t) < \frac{\sfd(s,t)}{\alpha}.
\]
Hence $(s,p)$ is an $\alpha$-navigation from $s$ to $t$, and taking all sources and scales gives an $\alpha$-navigable graph.

The two approximation guarantees use two closely related but different choices of global representatives. For the constant-factor guarantee, we use separated covers: at each scale, the selected centers are pairwise farther than the ball radius. This separation is a coarse, scale-by-scale analogue of the ordered-greedy packing condition above. For each fixed source $s$, this separation allows us to map
the selected centers at each scale to outgoing edges of $s$
in a sparsest $(2\alpha+O(1))$-navigable graph. The map is
injective at each scale, and each edge is used by these maps
at only constantly many scales.

For the logarithmic guarantee, we instead use the standard
greedy Set Cover algorithm over all radius-$\Theta(r_i/\alpha)$
balls, allowing the chosen center to be a point that has
already been covered. This variant does not have the separation
property used in the constant-factor analysis. For each source
and scale, the clusters whose centers lie in the annulus define
a local set-cover instance. We show that the number of selected
centers is at most $O(\log n)$ times the optimum of this local
instance. We then bound each local optimum by the number of
outgoing edges of the source in a sparsest
$(\alpha+O(1))$-navigable graph that navigate to points in the
local instance. Each such edge is counted at only constantly
many scales.

Finally, to obtain the simultaneous statement of \Cref{thm:alg}, we run both constructions and return the graph with fewer edges. Since both candidate graphs are $\alpha$-navigable, the returned graph is $\alpha$-navigable; because its size is no larger than either candidate, it inherits both bicriteria bounds.

%% file: alg.tex
\section{Bicriteria Approximation in Near-Quadratic Time.}\label{sec:n2-bicriteria}

\subsection{Algorithm Description.}
Let $(P,\sfd)$ be a finite metric space on $n$ points and let
$\alpha \ge 1$ be a fixed parameter. By rescaling all
distances by a common positive factor, we assume that the minimum
nonzero pairwise distance is $1+1/\alpha$. Consequently, all nonzero
pairwise distances lie in the range
$ \left[1+\frac{1}{\alpha}, \left(1+\frac{1}{\alpha}\right)\Delta\right]$, 
where $\Delta$ is the aspect ratio of the metric. For any $s\in P$ and radius parameter $r>0$, define
\[
    B(s,r)=\{u\in P\mid \sfd(s,u)\le r\}.
\]
Let $\delta=1/\alpha$, let $L=\ceil{\log_{1+\delta}\Delta}$, and define $r_i=(1+\delta)^i$ for every $0\le i\le L+1$.
Notice that $r_0=1$ and $r_{L+1}\ge (1+\delta)\Delta$. Hence, for every pair of distinct points $s,t\in P$, there exists an index $0\le i\le L$ such that $ r_i<\sfd(s,t)\le r_{i+1}.$
We use the following annulus around a source point $s$ at level $i$:
\[
A_i(s)=\left\{v\in P\;:\; r_i-\frac{r_i}{\alpha+2}<\sfd(s,v)\le r_{i+1}+\frac{r_i}{\alpha+2}\right\}.
\]

We use two different ball-cover routines. The first one is a separated cover and is used for the constant-factor bicriteria approximation. The second one is the usual greedy Set Cover algorithm over metric balls and is used for the logarithmic bicriteria approximation.

\paragraph*{Separated ball cover.}
Given a radius $r>0$, the following procedure partitions $P$ into clusters $C$ with centers $p$ satisfying $C\subseteq B(p,r)$. The centers are pairwise separated by distance strictly larger than $r$.

\begin{algorithm2e}[!htb]
\caption{SeparatedBallCover($P,\sfd,r$)}
\label{alg:separated-ball-cover}
\KwIn{metric $(P,\sfd)$, radius $r>0$}
\KwOut{a collection $\mathcal{N}$ of pairs $(C,p)$, where the sets $C$ partition $P$ and $C\subseteq B(p,r)$}
$\mathcal{N}\gets \emptyset$, $U\gets P$\;
\While{$U\neq\varnothing$}{
    Choose an arbitrary point $p\in U$\;
    $C\gets U\cap B(p,r)$, $U\gets U\setminus C$\;
    Add $(C,p)$ to $\mathcal{N}$\;
}
\Return{$\mathcal{N}$}\;
\end{algorithm2e}

For every pair $(C,p)$ returned by \Cref{alg:separated-ball-cover}, we have $C\subseteq B(p,r)$. Moreover, if $p$ is chosen before another center $p'$, then $p'$ is not covered by $B(p,r)$, and hence $\sfd(p,p')>r$. Thus all centers are pairwise separated by distance more than $r$.

\paragraph*{Greedy ball cover.}
For the logarithmic approximation, we instead run the standard greedy Set Cover algorithm where the universe is $P$ and the available sets are all radius-$r$ balls centered at points of $P$. Importantly, the center chosen at a step is allowed to be a point that has already been covered.

\begin{algorithm2e}[!htb]
\caption{GreedyBallCover($P,\sfd,r$)}
\label{alg:greedy-ball-cover}
\KwIn{metric $(P,\sfd)$, radius $r>0$}
\KwOut{a collection $\mathcal{G}$ of pairs $(C,p)$, where the sets $C$ partition $P$ and $C\subseteq B(p,r)$}
$\mathcal{G}\gets \emptyset$, $U\gets P$\;
\While{$U\neq\varnothing$}{
    Choose $p\in P$ maximizing $|U\cap B(p,r)|$\;
    $C\gets U\cap B(p,r)$, $U\gets U\setminus C$\;
    Add $(C,p)$ to $\mathcal{G}$\;
}
\Return{$\mathcal{G}$}\;
\end{algorithm2e}

The sets $C$ returned by \Cref{alg:greedy-ball-cover} form a partition of $P$, and each such set is contained in the radius-$r$ ball around its chosen center. Unlike \Cref{alg:separated-ball-cover}, the center $p$ need not belong to the currently uncovered set $U$ and need not belong to the cluster $C$.

\paragraph*{The constant-approximation construction.}
At every level $i$, we compute a separated ball cover with radius $r_i/(\alpha+2)$. For each source $s$, we add edges to the centers that lie in the annulus $A_i(s)$.

\begin{algorithm2e}[!htb]
\caption{ConstantApproxNavigableGraph($P,\sfd,\alpha$)}
\label{alg:constant-navigable-graph}
\KwIn{metric $(P,\sfd)$, parameter $\alpha\ge 1$}
\KwOut{navigable graph $G=(P,E)$}
$\delta\gets 1/\alpha$, $L\gets \ceil{\log_{1+\delta}\Delta}$\;
Define $r_i=(1+\delta)^i$ for every $0\le i\le L+1$\;
\ForEach{$0\le i\le L$}{
    $\rho_i\gets r_i/(\alpha+2)$\;
    $\mathcal{N}_i\gets \textsc{SeparatedBallCover}(P,\sfd,\rho_i)$\;
}
\ForEach{$s\in P$}{
    $E_s\gets\emptyset$\;
    \ForEach{$0\le i\le L$}{
        Compute $A_i(s)\gets\left\{v\in P: r_i-\rho_i<\sfd(s,v)\le r_{i+1}+\rho_i\right\}$\;
        $E_s^i\gets\{(s,p)\mid (C,p)\in\mathcal{N}_i,\ p\in A_i(s)\}$\;
        $E_s\gets E_s\cup E_s^i$\;
    }
}
\Return{$(P,\bigcup_{s\in P}E_s)$}\;
\end{algorithm2e}

\paragraph*{The logarithmic-approximation construction.}
The logarithmic construction is identical except that the level covers are produced by the greedy Set Cover routine.

\begin{algorithm2e}[!htb]
\caption{LogApproxNavigableGraph($P,\sfd,\alpha$)}
\label{alg:log-navigable-graph}
\KwIn{metric $(P,\sfd)$, parameter $\alpha\ge 1$}
\KwOut{navigable graph $G=(P,E)$}
$\delta\gets 1/\alpha$, $L\gets \ceil{\log_{1+\delta}\Delta}$\;
Define $r_i=(1+\delta)^i$ for every $0\le i\le L+1$\;
\ForEach{$0\le i\le L$}{
    $\rho_i\gets r_i/(\alpha+2)$\;
    $\mathcal{G}_i\gets \textsc{GreedyBallCover}(P,\sfd,\rho_i)$\;
}
\ForEach{$s\in P$}{
    $E_s\gets\emptyset$\;
    \ForEach{$0\le i\le L$}{
        Compute $A_i(s)\gets\left\{v\in P: r_i-\rho_i<\sfd(s,v)\le r_{i+1}+\rho_i\right\}$\;
        $E_s^i\gets\{(s,p)\mid (C,p)\in\mathcal{G}_i,\ p\in A_i(s)\}$\;
        $E_s\gets E_s\cup E_s^i$\;
    }
}
\Return{$(P,\bigcup_{s\in P}E_s)$}\;
\end{algorithm2e}

\subsection{Running Time.}
Both ball-cover routines can be implemented in $\tilde{O}(n^2)$ time for a fixed radius. First compute all radius-$r$ balls. During the covering process, maintain the number of uncovered points in each ball. Whenever a cluster $C$ is removed from the uncovered set, update the counters of all balls containing each point of $C$; this costs $\tilde{O}(n|C|)$ time for the iteration and $\tilde{O}(n^2)$ time over all iterations. A priority queue with lazy updates gives the required maximum-coverage choice in \Cref{alg:greedy-ball-cover}.

There are $L+1=O(\alpha\log\Delta)$ levels. For either graph construction, computing all covers therefore takes $\tilde{O}(\alpha n^2\log\Delta)$ time. Once the covers are computed, all edge sets can be constructed by checking, for every level $i$, every source $s$, and every selected center $p$, whether $p\in A_i(s)$. This takes $O(n^2L)$ time and also bounds the output size. Thus both \Cref{alg:constant-navigable-graph} and \Cref{alg:log-navigable-graph} run in $\tilde{O}(\alpha n^2\log\Delta)$ time, which is $\tilde{O}(n^2\log\Delta)$ for constant $\alpha$.

\subsection{Proof of Correctness.}

\begin{lemma}[Ball Coverage]\label{lem:shortcut-ball}
Let $(P,\sfd)$ be a metric space and let $\alpha\ge 1$. For any $u,c,v\in P$, if
\[
    \sfd(c,v)<\frac{\sfd(c,u)}{\alpha+1},
\]
then the directed edge $(u,c)$ is an $\alpha$-navigation from $u$ to $v$, i.e.,
\[
    \sfd(c,v)<\frac{\sfd(u,v)}{\alpha}.
\]
\end{lemma}

\begin{proof}
By triangle inequality,
\[
    \sfd(u,v)\ge \sfd(c,u)-\sfd(c,v).
\]
Thus it suffices to prove
\[
    \sfd(c,v)<\frac{\sfd(c,u)-\sfd(c,v)}{\alpha},
\]
which is equivalent to $(\alpha+1)\sfd(c,v)<\sfd(c,u)$. This is exactly the assumption.
\end{proof}

\begin{lemma}\label{lem:both-correct}
Both algorithms (\Cref{alg:constant-navigable-graph,alg:log-navigable-graph}) compute $\alpha$-navigable graphs for $(P,\sfd)$.
\end{lemma}

\begin{proof}
The proof only uses the following common property of the level-$i$ cover: its clusters form a partition of $P$, and each cluster $C$ with center $p$ satisfies $C\subseteq B(p,r_i/(\alpha+2))$.

Fix one of the two algorithms, a source point $s\in P$, and a target point $t\in P$ with $s\ne t$. Let $i$ be an index such that $r_i<\sfd(s,t)\le r_{i+1}$; such an index exists by the scaling assumption. Since the level-$i$ clusters cover $P$, the point $t$ lies in some cluster $C$ with center $p$, and hence
\[
    \sfd(t,p)\le \frac{r_i}{\alpha+2}.
\]
By triangle inequality,
\[
    \sfd(s,p)\le \sfd(s,t)+\sfd(t,p)\le r_{i+1}+\frac{r_i}{\alpha+2},
\]
and
\[
    \sfd(s,p)\ge \sfd(s,t)-\sfd(t,p)> r_i-\frac{r_i}{\alpha+2}.
\]
Thus $p\in A_i(s)$, so the algorithm adds the edge $(s,p)$.

Finally,
\[
    \frac{\sfd(p,s)}{\alpha+1}
    > \frac{1}{\alpha+1}\left(r_i-\frac{r_i}{\alpha+2}\right)
    = \frac{r_i}{\alpha+2}
    \ge \sfd(p,t).
\]
Thus $\sfd(p,t)<\sfd(p,s)/(\alpha+1)$. By \Cref{lem:shortcut-ball}, the edge $(s,p)$ is an $\alpha$-navigation from $s$ to $t$. Since this holds for every distinct pair $s,t\in P$, the output graph is $\alpha$-navigable.
\end{proof}

\subsection{Constant Bicriteria Approximation.}

\begin{lemma}\label{lem:constant-approx}
Let $(P,\sfd)$ be a metric space of $n$ points and let $\alpha\ge 1$. Let $G=(P,E)$ be the graph returned by \Cref{alg:constant-navigable-graph}. Then $G$ is an $(\alpha,\beta)$-bicriteria $5$-approximation, where
\[
    \beta=2(1+\delta)(\alpha+2)+2=2\alpha+8+\frac4\alpha.
\]
\end{lemma}

\begin{proof}
Let $G^\star=(P,E^\star)$ be a sparsest $\beta$-navigable graph, and fix a source vertex $s\in P$. Write $E_s^\star=\{(s,x)\in E^\star\}$ for the out-edges of $s$ in $G^\star$.

Since $E_s=\bigcup_{i=0}^L E_s^i$, we have
$|E_s|\le\sum_{i=0}^L|E_s^i|$.
To bound this sum, we consider all pairs $(i,p)$ such that
$(s,p)\in E_s^i$; the same center $p$ is counted separately
at different levels. For each such pair, the
$\beta$-navigability of $G^\star$ guarantees an edge
$(s,x)\in E_s^\star$ such that
\[
    \sfd(x,p)<\frac{\sfd(s,p)}{\beta}.
\]
Choose one such edge and define $\phi_s(i,p)=(s,x)$.

\begin{claim}\label{claim:const-injective}
For each fixed level $i$, the map $p\mapsto\phi_s(i,p)$
is injective on $\{p\in P\mid(s,p)\in E_s^i\}$.
\end{claim}

\begin{proof}[Proof of claim]
Suppose that two distinct centers $p,p'\in A_i(s)$ at level $i$ are both mapped to the same edge $(s,x)\in E_s^\star$. By triangle inequality,
\[
    \sfd(p,p')\le \sfd(p,x)+\sfd(x,p')
    < \frac{\sfd(s,p)}{\beta}+\frac{\sfd(s,p')}{\beta}.
\]
Since $p,p'\in A_i(s)$,
\[
    \sfd(s,p),\sfd(s,p')\le r_{i+1}+\frac{r_i}{\alpha+2},
\]
and therefore
\[
    \sfd(p,p')<\frac{2\left((1+\delta)+\frac{1}{\alpha+2}\right)r_i}{\beta}.
\]
On the other hand, the centers returned by \Cref{alg:separated-ball-cover} at level $i$ are separated by distance more than $r_i/(\alpha+2)$. Hence
\[
    \frac{r_i}{\alpha+2}<\sfd(p,p')<\frac{2\left((1+\delta)+\frac{1}{\alpha+2}\right)r_i}{\beta},
\]
which implies
\[
    \beta<2(1+\delta)(\alpha+2)+2.
\]
This contradicts the definition of $\beta$.
\end{proof}

\begin{claim}\label{claim:const-overlap}
For every edge $(s,x)\in E_s^\star$, there are at most
$5$ indices $i$ for which some $p$ with $(s,p)\in E_s^i$
satisfies $\phi_s(i,p)=(s,x)$.
\end{claim}

\begin{proof}[Proof of claim]
Suppose that $p\in A_i(s)$ is a level-$i$ center and $\phi_s(i,p)=(s,x)$. By definition of $\phi_s$,
\[
    \sfd(p,x)<\frac{\sfd(s,p)}{\beta}.
\]
Using $p\in A_i(s)$ and triangle inequality,
\[
    \sfd(s,x)<\left(1+\frac1\beta\right)\left(r_{i+1}+\frac{r_i}{\alpha+2}\right),
\]
and
\[
    \sfd(s,x)>\left(1-\frac1\beta\right)\left(r_i-\frac{r_i}{\alpha+2}\right)
    =r_i\cdot\frac{\alpha+1}{\alpha+2}\left(1-\frac1\beta\right).
\]
Thus, for a fixed edge $(s,x)$, the relevant values of $r_i$ lie in an interval whose ratio is at most
\begin{align*}
R_\beta
&=\frac{\left((1+\delta)+\frac1{\alpha+2}\right)\left(1+\frac1\beta\right)}{\frac{\alpha+1}{\alpha+2}\left(1-\frac1\beta\right)} \\
&=\frac{\alpha^2+4\alpha+2}{\alpha(\alpha+1)}\cdot
  \frac{2\alpha+9+\frac4\alpha}{2\alpha+7+\frac4\alpha}\\
&\leq \brac{1 + \frac{1}{\alpha}}^3\cdot \brac{1 + \frac{1}{\alpha}}
\le \left(1+\frac1\alpha\right)^4.
\end{align*}
Since consecutive values of $r_i$ differ by a factor of $1+1/\alpha$, there are at most $1+\log_{1+1/\alpha}R_\beta\le 5$ such levels.
\end{proof}

By \Cref{claim:const-injective}, for each fixed level $i$,
at most one pair $(i,p)$ is mapped to any given edge of
$E_s^\star$. By \Cref{claim:const-overlap}, such pairs occur
at no more than $5$ levels for each edge. Therefore,
\[
    |E_s|\le \sum_{i=0}^L |E_s^i|\le 5|E_s^\star|.
\]
Summing over all sources $s\in P$ gives $|E|\le 5|E^\star|$.
\end{proof}

\subsection{Logarithmic Bicriteria Approximation.}

\begin{lemma}\label{lem:log-approx}
Let $(P,\sfd)$ be a metric space of $n$ points and let $\alpha\ge 1$. Let $G=(P,E)$ be the graph returned by \Cref{alg:log-navigable-graph}. Then $G$ is an $(\alpha,\beta')$-bicriteria $8(\ln n+1)$-approximation, where
\[
    \beta'=(1+\delta)(\alpha+2)+2=\alpha+5+\frac2\alpha.
\]
\end{lemma}

\begin{proof}
Let $G^\star=(P,E^\star)$ be a sparsest $\beta'$-navigable graph. Fix a source point $s\in P$ and a level $i$. Write $E_s^\star=\{(s,x)\in E^\star\}$ for the outgoing edges of $s$ in $G^\star$. Write $\rho_i=r_i/(\alpha+2)$, and let $\mathcal G_i=\{(C_j,p_j)\}_j$ be the greedy ball cover at level $i$ in the order in which the clusters are produced.

Define
\[
    \mathcal U_{i,s}=\bigcup\{C_j\mid (C_j,p_j)\in\mathcal G_i,\ p_j\in A_i(s)\}.
\]
Let $\mathcal F_{i,s}=\{B(c,\rho_i)\cap\mathcal U_{i,s}\mid c\in P\}$, and let $\mathrm{OPT}_{i,s}$ be the optimum set-cover value for the instance $(\mathcal U_{i,s},\mathcal F_{i,s})$.

\begin{claim}\label{claim:greedy-local-bound}
For every source $s$ and level $i$,
\[
    |E_s^i|\le (\ln n+1)\mathrm{OPT}_{i,s}.
\]
\end{claim}

\begin{proof}[Proof of claim]
Write the output of \Cref{alg:greedy-ball-cover} at level $i$ as
\[
    (C_1,p_1),(C_2,p_2),\ldots,(C_m,p_m),
\]
in the order in which the clusters are produced. Let $U^{(j)}$ be the set of points that remain uncovered immediately before the $j$-th greedy step. Thus
\[
    C_j=U^{(j)}\cap B(p_j,\rho_i),
\]
and the greedy choice of $p_j$ implies that, for every center $c\in P$,
\begin{equation}\label{eq:global-greedy-choice}
    |C_j|=|U^{(j)}\cap B(p_j,\rho_i)|
    \ge |U^{(j)}\cap B(c,\rho_i)|.
\end{equation}
Define 
\(J_{i,s}=\{j\in[m]\mid p_j\in A_i(s)\}.\)
Then $E_s^i$ contains exactly one edge $(s,p_j)$ for each $j\in J_{i,s}$, and
\[
    \mathcal U_{i,s}=\bigcup_{j\in J_{i,s}} C_j.
\]

For every point $v\in C_j$, define its weight by
\[
    y_v=\frac1{|C_j|}.
\]
Thus each greedy step assigns total weight $1$ to the
points newly covered at that step.
For each $j\in J_{i,s}$, the total weight of the points
in $C_j$ is
\[
    \sum_{v\in C_j} y_v=1.
\]
Therefore
\begin{equation}\label{eq:local-charge-objective}
    \sum_{v\in\mathcal U_{i,s}} y_v
    =\sum_{j\in J_{i,s}}\sum_{v\in C_j}y_v
    =|J_{i,s}|=|E_s^i|.
\end{equation}
For this fixed level $i$, the weights $y_v$ are determined
by the greedy Set Cover execution on all of $P$ and do not
depend on the source $s$. We use these same weights to bound
$|E_s^i|$ for every local instance
$(\mathcal U_{i,s},\mathcal F_{i,s})$.

We next bound the total weight in any radius-$\rho_i$ ball:
for every $c\in P$,
\begin{equation}\label{eq:ball-charge-bound}
    \sum_{v\in B(c,\rho_i)} y_v \le H_n \le \ln n+1,
\end{equation}
where $H_n=1+1/2+\cdots+1/n$. Fix such a ball $Q=B(c,\rho_i)$. Consider only the greedy steps that cover at least one point of $Q$. Immediately before such a step $j$, let
\[
    k_j=|Q\cap U^{(j)}|
\]
be the number of still-uncovered points of $Q$, and let
\[
    q_j=|Q\cap C_j|
\]
be the number of points of $Q$ covered at this step. Since $Q$ itself is one of the available radius-$\rho_i$ balls, the candidate center $c$ would cover all $k_j$ currently uncovered points of $Q$. Hence, by \eqref{eq:global-greedy-choice},
\[
    |C_j|\ge k_j.
\]
The total weight of the $q_j$ points of $Q$ covered
at step $j$ is
\[
    \frac{q_j}{|C_j|}\le \frac{q_j}{k_j}.
\]
If the nonzero values of $k_j$ along these steps are denoted by
\[
    k^{(1)} > k^{(2)} > \cdots > k^{(a)} > 0,
\]
and $q^{(h)}=k^{(h)}-k^{(h+1)}$ with $k^{(a+1)}=0$, then
\[
    \sum_{v\in B(c,\rho_i)} y_v\leq \sum_{h=1}^a \frac{q^{(h)}}{k^{(h)}}
    \le \sum_{h=1}^a \sum_{\ell=k^{(h+1)}+1}^{k^{(h)}} \frac1\ell
    = H_{k^{(1)}}
    \le H_n.
\]
This proves \eqref{eq:ball-charge-bound}.

Now let $\mathcal B^\textnormal{opt}$ be an optimal cover of $\mathcal U_{i,s}$ by sets from $\mathcal F_{i,s}$, so $|\mathcal B^\textnormal{opt}|=\mathrm{OPT}_{i,s}$. Each set in $\mathcal B^\textnormal{opt}$ has the form $B(c,\rho_i)\cap\mathcal U_{i,s}$ for some $c\in P$. Since $y_v\ge0$ for every $v\in P$ and
$\mathcal B^\textnormal{opt}$ covers $\mathcal U_{i,s}$,
\begin{align*}
    |E_s^i|
    &=\sum_{v\in\mathcal U_{i,s}} y_v &&\text{by \eqref{eq:local-charge-objective}}\\
    &\le \sum_{B(c,\rho_i)\cap\mathcal U_{i,s}\in\mathcal B^\textnormal{opt}}
        \sum_{v\in B(c,\rho_i)\cap\mathcal U_{i,s}} y_v \\
    &\le \sum_{B(c,\rho_i)\cap\mathcal U_{i,s}\in\mathcal B^\textnormal{opt}}
        \sum_{v\in B(c,\rho_i)} y_v \\
    &\le H_n\cdot |\mathcal B^\textnormal{opt}|
    \le (\ln n+1)\mathrm{OPT}_{i,s},
\end{align*}
where the penultimate inequality uses \eqref{eq:ball-charge-bound}. This proves the claim.
\end{proof}

\begin{claim}\label{claim:nav-to-cover}
If an edge $(s,x)\in E^\star$ is a $\beta'$-navigation from $s$ to a point $v\in\mathcal U_{i,s}$, then $v\in B(x,\rho_i)$.
\end{claim}

\begin{proof}[Proof of claim]
Since $v\in\mathcal U_{i,s}$, there exists a pair $(C,p)\in\mathcal G_i$ such that $v\in C$ and $p\in A_i(s)$. Hence $\sfd(p,v)\le\rho_i$. By triangle inequality,
\[
    \sfd(s,v)\le \sfd(s,p)+\sfd(p,v)
    \le r_{i+1}+2\rho_i
    =r_i\left(1+\delta+\frac{2}{\alpha+2}\right).
\]
If $(s,x)$ is a $\beta'$-navigation from $s$ to $v$, then
\[
    \sfd(x,v)<\frac{\sfd(s,v)}{\beta'}
    \le \frac{r_i\left(1+\delta+\frac{2}{\alpha+2}\right)}{(1+\delta)(\alpha+2)+2}
    =\frac{r_i}{\alpha+2}=\rho_i.
\]
Thus $v\in B(x,\rho_i)$.
\end{proof}

For each level $i$, define
\[
E^\star_{s,i}=\left\{(s,x)\in E_s^\star\;\middle|\;\exists v\in\mathcal U_{i,s}\text{ such that }\sfd(x,v)<\frac{\sfd(s,v)}{\beta'}\right\}.
\]
By the $\beta'$-navigability of $G^\star$ and \Cref{claim:nav-to-cover}, the balls $\{B(x,\rho_i)\cap\mathcal U_{i,s}\mid (s,x)\in E^\star_{s,i}\}$ cover $\mathcal U_{i,s}$. Therefore,
\[
    \mathrm{OPT}_{i,s}\le |E^\star_{s,i}|.
\]

It remains to show that each edge of $E_s^\star$ belongs to only constantly many sets $E^\star_{s,i}$. Suppose $(s,x)\in E^\star_{s,i}$. Then there exists $v\in\mathcal U_{i,s}$ such that $\sfd(x,v)<\sfd(s,v)/\beta'$. Choose a pair $(C,p)\in\mathcal G_i$ with $v\in C$ and $p\in A_i(s)$. Since $\sfd(p,v)\le\rho_i$,
\[
    \sfd(s,v)\le r_{i+1}+2\rho_i
    =r_i\left(1+\delta+\frac2{\alpha+2}\right),
\]
and
\[
    \sfd(s,v)> r_i-2\rho_i
    =r_i\left(1-\frac2{\alpha+2}\right).
\]
Consequently,
\[
    r_i\left(1-\frac2{\alpha+2}\right)\left(1-\frac1{\beta'}\right)
    <\sfd(s,x)
    <r_i\left(1+\delta+\frac2{\alpha+2}\right)\left(1+\frac1{\beta'}\right).
\]
Thus, for a fixed edge $(s,x)$, the relevant values of $r_i$ lie in an interval whose ratio is at most
\begin{align*}
R_{\beta'}
&=\frac{\left(1+\delta+\frac2{\alpha+2}\right)\left(1+\frac1{\beta'}\right)}{\left(1-\frac2{\alpha+2}\right)\left(1-\frac1{\beta'}\right)} \\
&=\frac{\alpha^2+5\alpha+2}{\alpha^2}\cdot
  \frac{\alpha+6+\frac2\alpha}{\alpha+4+\frac2\alpha}\\
&\leq \brac{1 + \frac{1}{\alpha}}^5\cdot \brac{1 + \frac{1}{\alpha}}^2 =\left(1+\frac1\alpha\right)^7.
\end{align*}
Since consecutive values of $r_i$ differ by a factor of $1+1/\alpha$, the edge $(s,x)$ can belong to at most $1+\log_{1+1/\alpha}R_{\beta'}\le 8$ sets $E^\star_{s,i}$. Hence
\[
    \sum_{i=0}^L \mathrm{OPT}_{i,s}
    \le \sum_{i=0}^L |E^\star_{s,i}|
    \le 8|E_s^\star|.
\]
Combining this with \Cref{claim:greedy-local-bound},
\[
    |E_s|\le \sum_{i=0}^L |E_s^i|
    \le (\ln n+1)\sum_{i=0}^L\mathrm{OPT}_{i,s}
    \le 8(\ln n+1)|E_s^\star|.
\]
Summing over all sources $s\in P$ gives $|E|\le 8(\ln n+1)|E^\star|$.
\end{proof}

\subsection{Combining the Two Constructions.}

\begin{proof}[Proof of \Cref{thm:alg}]
Run both \Cref{alg:constant-navigable-graph} and \Cref{alg:log-navigable-graph}, and return the graph with fewer edges. By \Cref{lem:both-correct}, both candidate graphs are $\alpha$-navigable, so the returned graph is also $\alpha$-navigable. If $G_{\mathrm{const}}$ is the graph returned by \Cref{alg:constant-navigable-graph} and $G_{\log}$ is the graph returned by \Cref{alg:log-navigable-graph}, then the final output $G$ satisfies
\[
    |E(G)|\le |E(G_{\mathrm{const}})|
    \quad\text{and}\quad
    |E(G)|\le |E(G_{\log})|.
\]
Therefore \Cref{lem:constant-approx} gives the $(\alpha,2\alpha+8+4/\alpha)$-bicriteria $5$-approximation guarantee, and \Cref{lem:log-approx} gives the $(\alpha,\alpha+5+2/\alpha)$-bicriteria $O(\ln n)$-approximation guarantee. The running time remains $\tilde O(\alpha n^2\log\Delta)$, which is $\tilde O(n^2\log\Delta)$ for constant $\alpha$.
\end{proof}

%% file: lbnd.tex
\section{Hardness Proof.}
In this section we prove \Cref{thm:reduction}. We reduce Set Cover to $(\alpha,2\alpha-\eta)$-bicriteria approximation of $\alpha$-navigable graphs on metric point sets.

\begin{lemma}[\cite{Dinur2013AnalyticalAT}]\label{lem:set-cover-approx}
	Let $(\mathcal{U}, \mathcal{F})$ be a set cover instance with $n$ elements and $m = n^{d}$ sets for some constant $d = O(1)$. It is NP-hard to compute a $\frac{1}{2}\ln n$-approximation of minimum set cover.
\end{lemma}

\begin{proof}[Proof of \Cref{thm:reduction}]
Fix constants $\alpha\ge 1$ and $\eta\in(0,1)$. Let $(\mathcal U,\mathcal F)$ be a hard Set Cover instance from \Cref{lem:set-cover-approx}, with $n$ elements and $m\le n^d$ sets for a constant $d=O(1)$. We may assume that the instance is feasible and that its optimum value $|C^\star|$ is at least $20$: if $|C^\star|<20$, then the optimum can be found by exhaustive search over all subsets of at most $19$ sets, which takes polynomial time because $d$ is a constant.

\paragraph*{Reduction from Set Cover to $(\alpha,2\alpha-\eta)$-Bicriteria Approximation.}
Let
\[
    K=\left\lceil \frac{2(m+n+1)}{\eta}\right\rceil.
\]
We create $K$ disjoint copies of the Set Cover instance, denoted by $(\mathcal U_i,\mathcal F_i)$ for $1\le i\le K$. We represent each set and each element by a distinct vertex, using the same notation for the vertex and the object it represents. For each copy $i$, we also introduce a unique hub point $c_i$, and write
\[
    P_i=\{c_i\}\cup \mathcal F_i\cup \mathcal U_i,
    \qquad
    P=\bigcup_{i=1}^K P_i.
\]
Thus the number of metric points is $N=|P|=K(m+n+1)$.

To define the metric $(P,\sfd)$, we construct an undirected weighted graph $H=(P,F,\mu)$ and let $\sfd(\cdot,\cdot)$ be the shortest-path metric of $H$. The edges and weights in $H$ are defined as follows; see \Cref{fig:lb-metric} for an illustration.
\begin{enumerate}[leftmargin=*]
    \item For each $i\in[K]$, if element $e\in\mathcal U_i$ is contained in set $S\in\mathcal F_i$, add an edge $(S,e)$ with weight $\mu(S,e)=1$.
    \item For each $i\in[K]$, add an edge $(c_i,S)$ with weight $\mu(c_i,S)=1$ for every $S\in\mathcal F_i$.
    \item For any two points $p\in P_i$ and $q\in P_j$ with $i\ne j$, add an edge $(p,q)$ with weight $\mu(p,q)=2\alpha$.
\end{enumerate}

\begin{figure}[t]
\centering
\input{./fig/lower_bound_metric_figure}
\caption{The weighted graph $H$ whose shortest-path metric is $(P,\sfd)$ in the lower-bound construction. Each $P_i$ consists of the hub $c_i$, the set vertices $\mathcal F_i$, and the element vertices $\mathcal U_i$. Solid edges have weight $1$: hub-to-set edges and set-element incidence edges. Dashed edges indicate that every pair of vertices in two different copies is connected by an edge of weight $2\alpha$.}
\label{fig:lb-metric}
\end{figure}
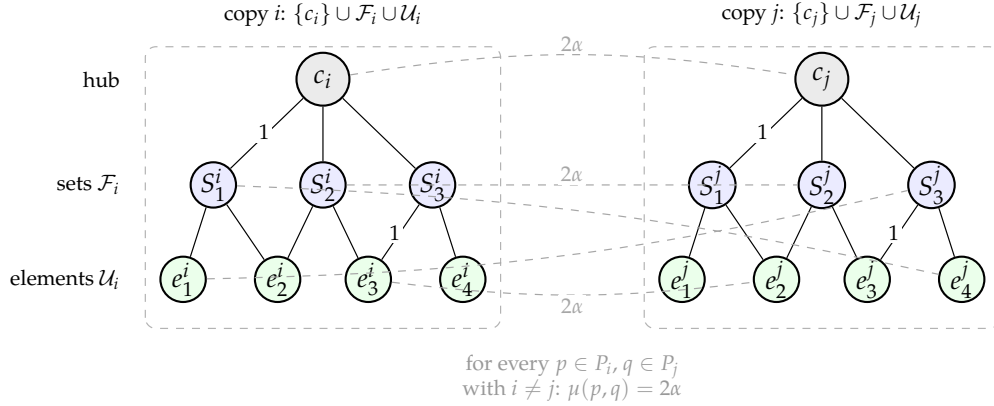

Now run a polynomial-time $(\alpha,2\alpha-\eta)$-bicriteria approximation algorithm on this metric and let $G=(P,E)$ be the resulting $\alpha$-navigable graph. Among all pairs $(s,j)$ with $s\in P_i$ and
$i,j\in[K]$, $i\ne j$, choose one that minimizes the number
of out-neighbors of $s$ in $\mathcal F_j\cup\mathcal U_j$. Denote by
\(
    V_{s,j}=\{p\in \mathcal F_j\cup\mathcal U_j : (s,p)\in E\}
\) the set of these out-neighbors.

\begin{lemma}\label{lem:valid-set-cover}
	The set $V_{s,j}\cap\mathcal F_j$ forms a set cover of $\mathcal U_j\setminus V_{s,j}$.
\end{lemma}
\begin{proof}
	Consider any element $x\in\mathcal U_j\setminus V_{s,j}$. Since $G$ is $\alpha$-navigable, there exists an edge $(s,p)\in E$ such that $\sfd(p,x)<\sfd(s,x)/\alpha$. Since $s$ and $x$ belong to different copies, $\sfd(s,x)=2\alpha$, and hence $\sfd(p,x)<2$. All edge weights in $H$ are at least $1$, and the only points at distance strictly less than $2$ from $x$ are the sets in $\mathcal F_j$ that contain $x$. Moreover, $x\notin V_{s,j}$, so $p\ne x$. Therefore $p\in V_{s,j}\cap\mathcal F_j$ and $p$ contains $x$.
\end{proof}

For every element $v\in\mathcal U_j\cap V_{s,j}$, choose an arbitrary set $S_v\in\mathcal F_j$ containing $v$, which exists because the original Set Cover instance is feasible. Define
\[
    C=(V_{s,j}\cap\mathcal F_j)\cup\{S_v: v\in\mathcal U_j\cap V_{s,j}\}.
\]
By \Cref{lem:valid-set-cover}, $C$ is a valid set cover for the $j$-th copy, and therefore for the original instance, with
\(
    |C|\le |V_{s,j}|.
\)

\paragraph*{Approximation Analysis.}
Let $C^\star$ be an optimal set cover solution for the original instance. For each index $i$, let $C^\star_i$ be a copy of $C^\star$ but on the set system $(\mathcal{U}_i, \mathcal{F}_i)$. We construct a $(2\alpha-\eta)$-navigable graph $G^\star=(P,E^\star)$ as follows.
\begin{itemize}[leftmargin=*]
    \item \textbf{Inter-copy edges:}
For every pair of distinct indices $i,j\in[K]$ and every
$s\in P_i$, add edges from $s$ to $c_j$ and to every
point in $C_j^\star$.
    \item \textbf{Intra-copy edges:} For every $i\in[K]$ and every pair of distinct points $s,t\in P_i$, add the directed edge $(s,t)$.
\end{itemize}

We first verify that $G^\star$ is $(2\alpha-\eta)$-navigable. If $s$ and $t$ belong to the same copy, then the direct edge $(s,t)$ gives distance $0$ to the target. Otherwise, suppose $s\in P_i$ and $t\in P_j$ for $j\ne i$. Then $\sfd(s,t)=2\alpha$. If $t=c_j$ or $t\in C_j^\star$, the direct edge to $t$ again gives distance $0$. If $t\in\mathcal F_j\setminus C^\star_j$, then the edge $(s,c_j)$ satisfies
\[
    \sfd(c_j,t)=1<\frac{2\alpha}{2\alpha-\eta}=\frac{\sfd(s,t)}{2\alpha-\eta}.
\]
Finally, if $t\in\mathcal U_j$, then some set $S\in C^\star_j$ contains $t$, and the edge $(s,S)$ satisfies
\[
    \sfd(S,t)=1<\frac{2\alpha}{2\alpha-\eta}=\frac{\sfd(s,t)}{2\alpha-\eta}.
\]
Thus $G^\star$ is $(2\alpha-\eta)$-navigable.

The total number of edges in $G^\star$ is
\[
    |E^\star|
    = K(K-1)(m+n+1)(|C^\star|+1)
      + K(m+n+1)(m+n).
\]
Suppose $G=(P,E)$ is an $(\alpha,2\alpha-\eta)$-bicriteria $\rho$-approximation to the sparsest $\alpha$-navigable graph, where
\(
    \rho=\frac{1}{10d}\ln N.
\)
Since $G^\star$ is feasible for the comparison parameter $2\alpha-\eta$, we have $|E|\le \rho |E^\star|$. Averaging over the $K(K-1)(m+n+1)$ pairs $(s,j)$ with
$s\in P_i$ and $i,j\in[K]$, $i\ne j$, shows that the
chosen pair $(s,j)$ satisfies
\begin{align*}
    |C|\le |V_{s,j}|
    &\le \frac{|E|}{K(K-1)(m+n+1)} \\
    &\le \rho\left(|C^\star|+1+\frac{m+n}{K-1}\right) \\
    &\le \rho\left(|C^\star|+1+\eta\right).
\end{align*}
Here the last inequality follows from the choice of $K$.

Since $m\le n^d$ and $\eta$ is a fixed constant, for all sufficiently large $n$ we have
\[
    N=K(m+n+1)\le \frac{100 n^{2d}}{\eta}.
\]
Therefore
\[
    \rho
    =\frac{1}{10d}\ln N
    \le \frac{1}{10d}\ln\left(\frac{100 n^{2d}}{\eta}\right)
    <0.21\ln n
\]
for all sufficiently large $n$. Also, by our assumption $|C^\star|\ge20$ and by $\eta<1$,
\[
    |C^\star|+1+\eta\le |C^\star|+2\le 1.1|C^\star|.
\]
Combining the above inequalities, we get
\[
    |C|<0.21\ln n\cdot 1.1|C^\star|<\frac12\ln n\cdot |C^\star|,
\]
contradicting \Cref{lem:set-cover-approx}. Thus, for $c_\eta=1/(10d)$, no polynomial-time $(\alpha,2\alpha-\eta)$-bicriteria $c_\eta\ln N$-approximation exists unless $\mathrm{P}=\mathrm{NP}$.
\end{proof}

%% file: fig/lower_bound_metric_figure.tex
\begin{tikzpicture}[
    x=1cm,y=1cm,
    hub/.style={circle,draw,thick,fill=black!8,minimum size=7mm,inner sep=0pt},
    setv/.style={circle,draw,thick,fill=blue!8,minimum size=6mm,inner sep=0pt},
    elemv/.style={circle,draw,thick,fill=green!8,minimum size=6mm,inner sep=0pt},
    localedge/.style={line width=0.45pt},
    crossedge/.style={line width=0.45pt,dashed,gray!70},
    copybox/.style={rounded corners,draw=gray!60,dashed,line width=0.45pt},
    every node/.style={font=\small}
]
    \node[hub] (ci) at (0,2.65) {$c_i$};
    \node[setv] (Fi1) at (-1.45,1.25) {$S^i_1$};
    \node[setv] (Fi2) at (0,1.25) {$S^i_2$};
    \node[setv] (Fi3) at (1.45,1.25) {$S^i_3$};
    \node[elemv] (Ui1) at (-1.85,0) {$e^i_1$};
    \node[elemv] (Ui2) at (-0.6,0) {$e^i_2$};
    \node[elemv] (Ui3) at (0.6,0) {$e^i_3$};
    \node[elemv] (Ui4) at (1.85,0) {$e^i_4$};

    \foreach \v in {Fi1,Fi2,Fi3} {
        \draw[localedge] (ci) -- (\v);
    }
    \draw[localedge] (Fi1) -- (Ui1);
    \draw[localedge] (Fi1) -- (Ui2);
    \draw[localedge] (Fi2) -- (Ui2);
    \draw[localedge] (Fi2) -- (Ui3);
    \draw[localedge] (Fi3) -- (Ui3);
    \draw[localedge] (Fi3) -- (Ui4);

    \draw[copybox] (-2.35,-0.65) rectangle (2.35,3.08);
    \node[font=\scriptsize] at (0,3.5) {copy $i$: $\{c_i\}\cup\mathcal F_i\cup\mathcal U_i$};
    \node[font=\scriptsize,anchor=east] at (-2.55,2.65) {hub};
    \node[font=\scriptsize,anchor=east] at (-2.55,1.25) {sets $\mathcal F_i$};
    \node[font=\scriptsize,anchor=east] at (-2.55,0) {elements $\mathcal U_i$};

    \node[hub] (cj) at (6.6,2.65) {$c_j$};
    \node[setv] (Fj1) at (5.15,1.25) {$S^j_1$};
    \node[setv] (Fj2) at (6.6,1.25) {$S^j_2$};
    \node[setv] (Fj3) at (8.05,1.25) {$S^j_3$};
    \node[elemv] (Uj1) at (4.75,0) {$e^j_1$};
    \node[elemv] (Uj2) at (6.0,0) {$e^j_2$};
    \node[elemv] (Uj3) at (7.2,0) {$e^j_3$};
    \node[elemv] (Uj4) at (8.45,0) {$e^j_4$};

    \foreach \v in {Fj1,Fj2,Fj3} {
        \draw[localedge] (cj) -- (\v);
    }
    \draw[localedge] (Fj1) -- (Uj1);
    \draw[localedge] (Fj1) -- (Uj2);
    \draw[localedge] (Fj2) -- (Uj2);
    \draw[localedge] (Fj2) -- (Uj3);
    \draw[localedge] (Fj3) -- (Uj3);
    \draw[localedge] (Fj3) -- (Uj4);

    \draw[copybox] (4.25,-0.65) rectangle (8.95,3.08);
    \node[font=\scriptsize] at (6.6,3.5) {copy $j$: $\{c_j\}\cup\mathcal F_j\cup\mathcal U_j$};

    \node[fill=white,inner sep=1pt,font=\scriptsize] at (-0.78,1.95) {$1$};
    \node[fill=white,inner sep=1pt,font=\scriptsize] at (0.95,0.62) {$1$};
    \node[fill=white,inner sep=1pt,font=\scriptsize] at (5.82,1.95) {$1$};
    \node[fill=white,inner sep=1pt,font=\scriptsize] at (7.55,0.62) {$1$};

    \draw[crossedge] (ci) to[bend left=9] node[above,fill=white,inner sep=1pt,font=\scriptsize] {$2\alpha$} (cj);
    \draw[crossedge] (Fi2) -- node[above,fill=white,inner sep=1pt,font=\scriptsize] {$2\alpha$} (Fj2);
    \draw[crossedge] (Ui3) to[bend right=7] node[below,fill=white,inner sep=1pt,font=\scriptsize] {$2\alpha$} (Uj2);
    \draw[crossedge] (Fi1) to[bend left=5] (Uj4);
    \draw[crossedge] (Ui1) to[bend right=6] (Fj3);

    \node[align=center,gray!80,font=\scriptsize] at (3.3,-1.3)
        {for every $p\in P_i$, $q\in P_j$ \\
         with $i\ne j$: $\mu(p,q)=2\alpha$};
\end{tikzpicture}

%% file: sections/ai_disclosure.tex
\section*{Disclosure of AI Use.}

\begin{flushleft}
\renewcommand{\arraystretch}{1.3}
\begin{tabular}{@{}l@{\qquad}c@{}}
\textbf{Activity} & \textbf{AI used} \\
Asking the research question & No \\
Coming up with the approach & No \\
Proof development & No \\
Writing and exposition & Yes \\
Checking for bugs & Yes \\
Other supporting tasks & No \\
\end{tabular}
\end{flushleft}

\begin{sloppypar}
\paragraph{Details.}
\begin{enumerate}
\item \emph{Asking the research question.}
The authors formulated the research question.

\item \emph{Coming up with the approach.}
The authors developed the overall conceptual approach.

\item \emph{Proof development.}
The authors developed all the main technical ideas underlying the proofs.

\item \emph{Writing and exposition.}
The authors wrote an initial draft of the manuscript and subsequently
used ChatGPT to help refine its exposition and presentation.

\item \emph{Checking for bugs.}
The authors used ChatGPT to check technical arguments. It helped
identify a bug in the analysis in an earlier version of Section~2.5
(Logarithmic Bicriteria Approximation).

\item \emph{Other supporting tasks.}
We report no additional substantial AI contributions to supporting tasks.
\end{enumerate}

The authors reviewed and finalized all mathematical claims, proofs,
and text, and take full responsibility for the correctness,
originality, exposition, and attribution of the final paper.
\end{sloppypar}

%% file: sections/acknowledgment.tex
\section*{Acknowledgment.}
Hongyi Chen and Tianyi Zhang are supported by the Fundamental and Interdisciplinary Disciplines Breakthrough Plan of the Ministry of Education of China (No. JYB2025XDXM118) and the ``111 Center'' (No. B26023).